\documentclass{llncs}
\usepackage{graphicx,amsfonts,subfigure,float}
\pdfoutput=1
\usepackage[T1]{fontenc}
\usepackage[title]{appendix}
\let\cleardoublepage\clearpage

\graphicspath{{./}{figures/}}

\newcommand{\Avatar}{$\textsc{Avatar}_{\textsc{Cbt}}$ }
\newcommand{\proofsketch}{\noindent \emph{Proof sketch: }}

\floatstyle{ruled}
\newfloat{algorithm}{th}{lop}
\floatname{algorithm}{Algorithm Sketch}

\begin{document}
\title{Avatar: A Time- and Space-Efficient Self-Stabilizing Overlay Network}
\author{Andrew Berns}
\institute{Department of Computer Science\\
University of Wisconsin-La Crosse\\
La Crosse, Wisconsin, USA\\
\email{aberns@uwlax.edu}}
\date{}
\maketitle

\begin{abstract}
Overlay networks present an interesting challenge for fault-tolerant computing.  Many overlay networks operate in dynamic environments (e.g. the Internet), where faults are frequent and widespread, and the number of processes in a system may be quite large.  Recently, self-stabilizing overlay networks have been presented as a method for managing this complexity.  \emph{Self-stabilizing overlay networks} promise that, starting from any weakly-connected configuration, a correct overlay network will eventually be built.  To date, this guarantee has come at a cost: nodes may either have high degree during the algorithm's execution, or the algorithm may take a long time to reach a legal configuration.  In this paper, we present the first self-stabilizing overlay network algorithm that does not incur this penalty.  Specifically, we (i) present a new locally-checkable overlay network based upon a binary search tree, and (ii) provide a randomized algorithm for self-stabilization that terminates in an expected polylogarithmic number of rounds \emph{and} increases a node's degree by only a polylogarithmic factor in expectation.
\end{abstract}

\section{Introduction}
Today's distributed systems are quite different than those from only a decade ago.  Pervasive network connectivity and an increase in the number of computational devices has ushered in an era of large-scale distributed systems which operate in highly-dynamic environments.  One type of distributed system that has gained popularity recently is the overlay network.  An \emph{overlay network} is a network where communication occurs over \emph{logical} links, where each logical link consists of zero or more \emph{physical} links.  The use of logical links allows the design of efficient logical topologies (e.g. topologies with low diameter and/or low degree) irrespective of the physical topology.  These topologies allow the construction of efficient data structures from large systems with arbitrary physical networks.

The dynamic nature of many overlay networks places extreme importance on the ability to handle a wide variety of faults.  \emph{Self-stabilization}, first presented by Dijkstra in 1974~\cite{dijkstra_ssorig_74}, is an elegant fault-tolerant paradigm promising that, after \emph{any} memory-corrupting transient fault, the system will eventually recover to the correct configuration.  \emph{Self-stabilizing overlay networks}, then, are logical networks that guarantee a correct topology will be restored after any transient memory corruption.

\subsection{Related Work}
Many overlay networks include a mechanism to tolerate at least a subset of possible faults.  For instance, the \textsc{Chord} overlay network~\cite{stoica_chord_01} defines a procedure for nodes to join the network efficiently.  The \textsc{Forgiving Graph}~\cite{hayes_forgiving_12} presents a self-healing overlay network which maintains connectivity while limiting degree increases and stretch despite periodic adversarial node insertions and deletions.

Self-stabilizing overlay networks are a relatively new area of overlay network research.  In 2007, Onus et al. presented the first silent self-stabilizing overlay network, which built a linear topology in linear (in the number of nodes) rounds~\cite{onus_linear_07}.  The first self-stabilizing overlay network with polylogarithmic convergence time was the \textsc{Skip+} graph, presented in 2009 by Jacob et al.~\cite{jacob_skipplus_09}.  Berns et al. presented a generic framework capable of building any locally-checkable overlay network, and proved that their result was near-optimal in terms of running time~\cite{berns_tcs_13}.

Current self-stabilizing overlay networks have suffered from one of two limitations.  First, some self-stabilizing overlay networks require a long time to reach the correct configuration.  For instance, the \textsc{ReChord} network~\cite{kniesburges_rechord_11} is a self-stabilizing variant of the \textsc{Chord} network, but requires $\mathcal{O}(n \log n)$ rounds to do so.  Conversely, some self-stabilizing overlay networks can converge quickly, but may require a large amount of space.  The \textsc{Skip+} network~\cite{jacob_skipplus_09}, for example, has a polylogarithmic convergence time, but may increase a node's degree to $\mathcal{O}(n)$ during convergence.  The Transitive Closure Framework of Berns et al.~\cite{berns_tcs_13} requires $\Theta(n)$ space.  To date, no work has achieved efficient convergence in terms of both time and space.

\subsection{Contributions}
In Section \ref{section:avatar}, we present \textsc{Avatar}, a generic locally checkable overlay network, and describe a specific ``instance'' of the network called \Avatar, based upon a binary search tree.  Section \ref{section:avatar_algo} presents a randomized self-stabilizing algorithm for creating the \Avatar network, as well as an analysis sketch demonstrating the algorithm is efficient in terms of both convergence time and space (using a new metric we call the \emph{degree expansion}).
%
%
%
%
%
%
%
%
%

\section{Model of Computation}
\label{section:model}
We model the distributed system as an undirected graph $G = (V,E)$, with nodes $V$ representing the processes of the system, and edges $E$ representing the communication links.  Each node $u$ is assigned an identifier from the function $\mathit{id}: V \rightarrow \mathbb{Z}^+$.  We assume each node stores $\mathit{id}(u)$ as immutable data.  Where clear from context, we will refer to a node $u$ by its identifier.

Each node $u \in V$ has a \emph{local state} $S(u)$ consisting of a set of variables and their values.  We assume all nodes also have access to a shared (immutable) random sequence $\Psi$.  A node $u$ can modify the value of its variables using \emph{actions} defined in the \emph{program} of $u$.  All nodes execute the same program.  We use a \emph{synchronous} model of computation, where in one \emph{round} each node executes its program and communicates with its neighbors.  We use the \emph{message passing} model of communication, where a node $u$ can communicate with a node $v$ in its \emph{neighborhood} $N(u)=\{v \in V:(u,v) \in E\}$ by sending node $v$ (called a \emph{neighbor}) a message.  A node can send unique messages to every neighbor in every round.  Messages sent to node $u$ in round $i$ are received by node $u$ at the start of round $(i+1)$.  We assume reliable and bounded capacity communication channels where a message is received if and only if it was sent in the previous round.

In the overlay network model, a node's neighborhood is part of its state, allowing a node to change its neighborhood with program actions.  Specifically, in a round $i$, a node $u$ can delete any edge incident upon it, or add an edge to any node $v$ which is currently distance 2 from $u$.  Specifically, let $G_i$ be the configuration in round $i$.  A node $u$ can (i) delete any edge $(u,v) \in E(G_i)$, resulting in $(u,v) \notin E(G_{i+1})$, and (ii) create the edge $(u,w)$ if $(u,v), (v,w) \in E(G_{i})$, resulting in $(u,w) \in E(G_{i+1})$.  We restrict edge additions to only those nodes at distance 2 to reflect the fact that only nodes at distance 2 share a common neighbor through which they can be ``connected''.  We assume that $v \in N(u) \Leftrightarrow (u,v) \in E$ -- that is, every neighbor of $u$ is known, and $u$ has no nodes in $N(u)$ which do not exist.  This can be achieved with the use of a ``heartbeat'' message sent in each round.

Our overlay network problem is to take a set of nodes $V$ and create a \emph{legal configuration}, where the legal configuration is defined by some predicate taken over the state of all nodes in $V$.  Since edges are state in an overlay network, the predicate for a legal configuration often includes the requirement that the topology matches a particular \emph{desired topology} $ON(V) = (V, E)$.  The \emph{self-stabilizing overlay network problem} is to design an algorithm $\mathcal{A}$ such that, when executed on nodes $V$ with arbitrary initial state in an arbitrary weakly-connected initial topology, the system reaches a legal configuration.  Furthermore, once the network is in a legal configuration, it remains in this legal configuration until an external fault perturbs the system.

Performance of an overlay network algorithm can be measured in terms of both time and space.  To analyze the worst-case performance, it is assumed that an \emph{adversary} creates the initial configuration using full knowledge of the nodes and algorithm (excluding the value of the shared random sequence $\Psi$).  The maximum number of rounds required for a legal configuration to be reached, taken over all possible initial configurations, is called the \emph{convergence time} of $\mathcal{A}$.  The number of incident edges during convergence is a major space consideration for overlay network algorithms as each incident edge requires memory and communication (for heartbeat messages).  To quantify this growth, we introduce the \emph{degree expansion}, which is, informally, the amount a node's degree may grow ``unnecessarily'' during convergence.  For a graph $G$ with node set $V$, let $\Delta_G$ be the maximum degree of nodes in $G$.  For a self-stabilizing algorithm $\mathcal{A}$ executing on $G$, let $\Delta_{\mathcal{A},G}$ be the maximum degree of any node from $V$ during execution of $\mathcal{A}$ beginning from configuration $G$.  We define degree expansion as follows.

\begin{definition}
The \emph{degree expansion of $\mathcal{A}$ on $G$}, denoted $\mathit{DegExp}_{\mathcal{A},G}$, is equal to $(\Delta_{\mathcal{A},G} / \max(\Delta_G, \Delta_{\mathit{ON}(\lambda)}))$.  Let the \emph{degree expansion of $\mathcal{A}$} be $\mathit{DegExp}_{\mathcal{A}} = \max_{G \in \mathcal{G}}(\mathit{DegExp}_{\mathcal{A},G})$
\end{definition}

An adversary could create an initial configuration where many edges are forwarded to a node in one round.  The degree expansion is meant to capture the degree growth ``caused'' by the algorithm itself, not the adversary.

We say that a self-stabilizing overlay network algorithm is \emph{silent} if and only if the algorithm brings the system to a configuration where the messages exchanged between nodes remains fixed until a fault perturbs the system~\cite{dolev_ssbook_00}.  Traditionally, these messages consist of the state of a node $u$.  In order for a silent self-stabilizing overlay network algorithm to exist given only this information, the legal configuration must be locally checkable.  An overlay network is \emph{locally checkable} if and only if each configuration which is not a legal configuration has at least one node (called a \emph{detector}) which detects that the configuration is not legal using only its state and the state of its neighbors, and all legal configurations have no detectors.

%
%
%
%
%
%
%
%
%
%
%

\section{The \textsc{Avatar} Network}
\label{section:avatar}
\subsection{\textsc{Avatar} Specification}
One of the challenges with creating silent self-stabilizing overlay network algorithms is in designing a topology that is locally checkable.  This is a non-trivial task as many popular overlay networks are in fact not locally checkable.  This is demonstrated in prior work.  The \textsc{Skip+} network~\cite{jacob_skipplus_09} was created as a locally-checkable variant of the \textsc{Skip} graph~\cite{aspnes_skipgraph_03}.  Similarly, the self-stabilizing \textsc{ReChord} network~\cite{kniesburges_rechord_11} required a variant of \textsc{Chord}~\cite{stoica_chord_01} for local checkability using real and virtual nodes.  Simplifying this network design task is the motivation for \textsc{Avatar}.  \textsc{Avatar} allows many different topologies to be ``simulated'' while ensuring local checkability, simplifying the network design task when creating stabilizing overlays.

A \emph{network embedding} $\Phi$ maps the node set of a \emph{guest network} $G_g = (V_g, E_g)$ onto the node set of a \emph{host network} $G_h = (V_h, E_h)$~\cite{leighton_book_1991}.  The \emph{dilation} of $\Phi$ is defined as the maximum distance between any two nodes $\Phi(u), \Phi(v) \in V_h$ such that $(u,v) \in E_g$.  The \textsc{Avatar} network, informally speaking, is an overlay network designed to realize a dilation-1 embedding for a guest network using the logical overlay links.  To do this \emph{and} ensure local checkability, the (host) overlay edges of \textsc{Avatar} consist of the successor and predecessor edges from a linearized graph (ensuring host nodes know which guest nodes should be embedded on them) as well as the overlay edges between host nodes necessary for neighboring guest nodes to also have neighboring hosts.

More formally, for any $N \in \mathbb{N}$, let $[N]$ be the set of nodes $\{0,1,\ldots,N-1\}$.  Let $\mathcal{F}$ be a family of graphs such that, for each $N \in \mathbb{N}$, there is exactly one graph $F_N \in \mathcal{F}$ with node set $[N]$.  We use $\mathcal{F}(N)$ to denote $F_N$.  We call $\mathcal{F}$ a \emph{full graph family}, capturing the notion that the family contains exactly one topology for each ``full'' set of nodes $[N]$ (relative to the identifiers).  For any $N \in \mathbb{N}$ and $V \subseteq [N]$, $\textsc{Avatar}_{\mathcal{F}}(N,V)$ is a network with node set $V$ that realizes a dilation-1 embedding of $F_N \in \mathcal{F}$.  The specific embedding is given below.  We also show that, when given $N$, \textsc{Avatar} is locally checkable ($N$ can be viewed as an upper bound on the number of nodes in the system).

\begin{definition}
Let $V \subseteq [N]$ be a node set $\{u_0, u_1, \ldots, u_{n-1}\}$, where $u_i < u_{i+1}$ for $0 \leq i < n-1$.  Let the \emph{range} of a node $u_i$ be $\mathit{range}(u_i) = [u_i, u_{i+1})$ for $0 < i < n-1$.  Let $\mathit{range}(u_0) = [0, u_1)$ and $\mathit{range}(u_{n-1}) = [u_{n-1}, N)$.  $\textsc{Avatar}_{\mathcal{F}}(N, V)$ is a graph with node set $V$ and edge set consisting of two edge types:
\begin{description}
\item[Type 1:] $\{(u_i, u_{i+1}) | i=0,\ldots,n-1\}$
\item[Type 2:] $\{(u_i, u_j) |  u_i \neq u_j \wedge \exists (a,b) \in E(F_N), a \in \mathit{range}(u_i) \wedge b \in \mathit{range}(u_j)\}$
\end{description}
\label{defn:avatar}
\end{definition}

\begin{theorem}
Let $\mathcal{F}$ be an arbitrary full graph family, and let $\textsc{Avatar}_{\mathcal{F}}(N, V)$ be an overlay network for some arbitrary $N$ and $V$, with all $u \in V$ having knowledge of $N$.  $\textsc{Avatar}_{\mathcal{F}}(N, V)$ is locally checkable.
\label{theorem:locally_checkable_mutable}
\end{theorem}
\proofsketch To prove this theorem, note each node can calculate its range using only its neighborhood.  As a node $u$ receives the state of its neighbors in each round, $u$ can also calculate the range of its neighbors.  This information is sufficient for each node $u$ to verify every neighbor $v \in N(u)$ is either from a type 1 or type 2 edge.  As all nodes know $N$ and there is exactly one $F_N \in \mathcal{F}$, all nodes can verify their type 1 and type 2 edges correctly map to the given network.

Interestingly, $\textsc{Avatar}_{\mathcal{F}}$ is locally checkable even when only $\mathcal{O}(\log n)$ bits of information are exchanged between neighbors.  Specifically, if every node shares with its neighbors (i) its identifier, and (ii) the identifier of its predecessor and successor, $\textsc{Avatar}_{\mathcal{F}}$ is locally checkable.

\subsection{The Full Graph Family $\textsc{Cbt}$}
\label{section:cbt}
Our goal is to create a self-stabilizing \textsc{Avatar} network which maintains low degree during stabilization and yet stabilizes quickly.  To this end, we selected a graph family of a simple data structure with constant degree and logarithmic diameter for all nodes: a binary search tree.  As we will demonstrate, not only does a complete binary search tree have low degree and diameter, but an embedding in \textsc{Avatar} of the binary search tree does as well.

More formally, consider a simple graph family based upon the \emph{complete binary search tree}.  We define the full graph family $\textsc{Cbt}$ by defining $\textsc{Cbt}(N)$ recursively in Definition \ref{defn:cbt}.

\begin{definition}
\label{defn:cbt}
For $a \leq b$, let $\textsc{Cbt}[a,b]$ be a binary tree rooted at $\mathit{r} = \lfloor (b+a)/2 \rfloor$.  Node $r$'s left cluster is $\textsc{Cbt}[a,r-1]$, and $r$'s right cluster is $\textsc{Cbt}[r+1,b]$.  If $a > b$, then $\textsc{Cbt}[a,b] = \bot$.  We define $\textsc{Cbt}(N) = \textsc{Cbt}[0,N-1]$.  Let the \emph{level} of a node $d$ in $\textsc{Cbt}[0,N-1]$ be the distance from $d$ to root $\lfloor N-1 / 2 \rfloor$.
\end{definition}

\subsubsection{Diameter and Maximum Degree of $\textsc{Avatar}_{\textsc{Cbt}}$}
All dilation-1 embeddings preserve the diameter of the guest network, meaning $\textsc{Avatar}_{\textsc{Cbt}}$ has $\mathcal{O}(\log N)$ diameter.  However, a node $v$ in our embedding may have a large $\Phi^{-1}(v)$ -- that is, many nodes from the guest network may map to a single node in the host network.  Surprisingly, however, the host nodes for $\textsc{Avatar}_{\textsc{Cbt}}$ have a small degree \emph{regardless of $\Phi^{-1}$}.  We sketch the proof for this result below.

\begin{theorem}
For any node set $V \subseteq [N]$, the maximum degree of any node $u \in V$ in $\textsc{Avatar}_{\textsc{Cbt}}(N, V)$ is at most $2 \cdot \log N + 2$.
\label{theorem:max_degree}
\end{theorem}
\proofsketch Consider $\Phi^{-1}(u)$, the subset of nodes from $[N]$ mapped to node $u$.  Let $[N]_j$ be the set of all nodes at level $j$ of $\textsc{Cbt}(N)$.  There are at most 2 nodes in $\Phi^{-1}(u) \cap [N]_j$ with a neighbor not in $\Phi^{-1}(u)$ -- that is, there are at most 2 edges from the range of a node $u$ to any other node outside this range for a particular level $j$ of the tree.  As there are only $\log N + 1$ levels, the total degree of any node in $\textsc{Avatar}_{\textsc{Cbt}}$ is at most $2 \cdot \log N + 2$.

\section{A Self-Stabilizing Algorithm}
\label{section:avatar_algo}
\subsection{Algorithm Overview}
At a high level, our self-stabilizing algorithm works on the same principle as the algorithm for constructing a minimum-weight spanning tree presented by Gallager, Humblet, and Spira~\cite{ghs_mst_83}.  The network is organized into disjoint clusters, each with a leader.  The cluster leaders coordinate the merging of clusters until only a single cluster remains, at which point the network is in a legal configuration.

Self-stabilizing overlay networks add an interesting aspect to this pattern.  To converge from an arbitrary weakly-connected configuration while limiting a node's degree increase requires coordination of merges, which requires either time (additional rounds) or bandwidth (additional edges).  In the overlay network model, we can increase both of these: we can add edges to the network and steps to our algorithm.  Our algorithm balances these aspects using the four components discussed below to achieve polylogarithmic convergence time and degree growth.
\begin{enumerate}
\item \textbf{Clustering:} As any weakly-connected initial configuration is allowable, we must ensure all nodes quickly join a cluster, as well as have a way to efficiently communicate within their cluster.  We define a cluster for \Avatar and present mechanisms for creating and communicating within clusters.
\item \textbf{Matching:} Progress comes from clusters in the system merging and moving towards a single-cluster configuration.  However, as we will show, merging clusters results in an $\mathcal{O}(\log^2 N)$ degree increase for each merging cluster.  To control this degree growth, we limit a cluster to merging with at most one other cluster at a time.  We determine which clusters should merge by creating a matching.  We introduce a mechanism to create sufficiently-many matchings on \emph{any} topology.  Our mechanism relies on the ability to add edges in the overlay network model.
\item \textbf{Merging:} Once two clusters are matched they can merge together into one.  Merging quickly requires sufficient ``bandwidth'' (in the form of edges) between two clusters.  However, to limit degree increases, these edges must be created carefully.  We present an algorithm for merging two clusters quickly while still limiting the number of additional edges that are created.
\item \textbf{Termination Detection:} Finally, to ensure our algorithm is silent, we define a simple mechanism for detecting when the legal configuration has been reached, allowing our algorithm to terminate.
\end{enumerate}

We discuss these components below, providing sketches of the algorithms and analysis.  The full algorithms and analysis can be found in the appendix, along with several clarifying figures demonstrating the merge process.

\subsection{Clustering}
\subsubsection{Defining a Cluster}
In the overlay network model, we can create clusters by defining both the nodes in the cluster, as well as the \emph{topology} of the cluster.  Our algorithm uses the following cluster definition.

\begin{definition}
Let $G$ be a graph with node set $V$.  A \emph{\textsc{Cbt} cluster} is a set of nodes $V' \subseteq V$ in graph $G$ such that $G[V']$, the subgraph of $G$ induced by $V'$, is $\textsc{Avatar}_{\textsc{Cbt}}(N,V')$.
\end{definition}

Notice our cluster can be thought of on two levels: on one level, it consists of an $N$-node guest \textsc{Cbt} network, while on the other level, it consists of host nodes $V'$.  We call the root of the guest \textsc{Cbt} network the \emph{root} of the cluster.  Figure \ref{fig:clusters_real} contains the (host) graph $G$ with two clusters -- one consisting of nodes in various shades of green, and one with nodes in shades of red.  The two (guest) \textsc{Cbt} networks corresponding to these clusters are given in Figure \ref{fig:clusters_virtual}.

\begin{figure}
\centering
   \includegraphics[scale=0.3]{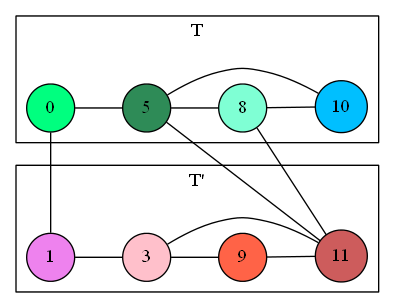}
   \caption{Host nodes of clusters $T$ (top) and $T'$ (bottom)}
   \label{fig:clusters_real}
\end{figure}
\begin{figure}
\centering
   \includegraphics[scale=0.27]{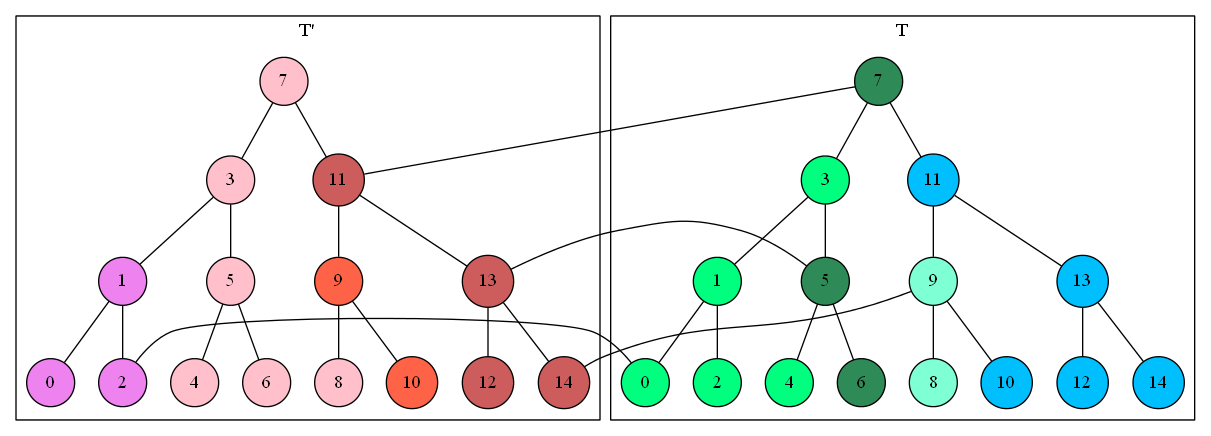}
   \caption{Guest Nodes for $T$ (right) and $T'$ (left)}
   \label{fig:clusters_virtual}
\end{figure}

To ensure all nodes are members of a cluster quickly, our \textsc{Cbt} clusters (from here on, simply clusters) must be locally checkable.  For local checkability, we add two variables to each node: a cluster identifier $cluster_u$ containing the identifier of the host of the root node in the cluster, and a cluster predecessor $clusterPred_u$ and successor $clusterSucc_u$, set to the closest identifiers in the subgraph induced by nodes with the same cluster identifier.  Let the \emph{cluster range of $u$} be the range of $u$ defined by the \emph{cluster} predecessor and successor.  We call a valid cluster a set of nodes $V'$ such that (i)the subgraph induced by $V'$ matches $\textsc{Avatar}_{\textsc{Cbt}}(N, V')$, and (ii) the legal range in $\textsc{Avatar}_{\textsc{Cbt}}(N, V')$ matches the \emph{cluster range} of $u$ in the configuration $G$.

Like $\textsc{Avatar}_{\textsc{Cbt}}(N, V')$, a cluster of nodes $V'$ is locally checkable.  The proof follows closely the proof that $\textsc{Avatar}_{\textsc{Cbt}}$ is locally checkable: nodes in a cluster $V'$ can calculate their cluster range, as well as the cluster range of their neighbors with matching cluster identifiers.  This allows them to check if their edges are from \Avatar$(N, V')$.  Furthermore, if the cluster identifier is invalid, at least one node can detect this.  We shall see later that there are some cases where a node is not a member of a cluster due to program actions (i.e. merge), and we will show this is also locally checkable.

In the self-stabilizing setting, there is no guarantee that each node begins execution belonging to a valid cluster.  To ensure this quickly, we define a ``reset'' operation nodes execute when a particular configuration is detected.  We say that a node $u$ has detected a \emph{reset fault} if $u$ detects (i) it is not a member of a cluster, (ii) this is not a state reachable from a ``legal'' merge (as we shall see, merges require coordination which allows nodes to differentiate invalid clusters caused from merging, and those from a fault), and (iii) it did not reset in the previous round.  When $u$ detects a reset fault, it ``resets'' to its own cluster of size 1.

\subsubsection{Intra-Cluster Communication}
\label{section:pfc}
Our algorithms rely on a systematic and reliable method of communication amongst nodes in the same cluster $T$.  We use a non-snap-stabilizing variant of the \emph{propagation of information with feedback and cleaning} ($\mathit{PFC}$) algorithm~\cite{bui_pfc_99}, which we ``simulate'' on the guest network $\textsc{Cbt}$ network for a cluster $T$ (denoted $\textsc{Cbt}_T(N)$).  The root node initiates a \emph{$\mathit{PFC}$ wave}, which (i) propagates information down the tree level-by-level until reaching the leaves, (ii) sends a feedback wave from the leaves to the root, passing along any requested feedback information, and (iii) prepares all nodes for another $\mathit{PFC}$ wave.

To allow the host network to simulate the $PFC$ algorithm, it is sufficient to append the ``level'' of the sender in the guest network to each message.  For instance, imagine the guest root wishes to initiate a $PFC$ wave by sending the message $m$ to its two children.  The host of the root will send the message $(m, 0)$ to the (at most two) hosts of the children.

\subsubsection{Analysis of Cluster Creation and Communication}
Below we analyze the performance of our cluster creation and communication mechanisms.

\begin{lemma}
Every node $u$ will be a member of a cluster in $\mathcal{O}(\log N)$ rounds.
\label{lemma:reset}
\end{lemma}
\proofsketch The lemma holds easily for nodes that are members of a cluster initially.  Consider a node $u$ that is not a member of a cluster.  If a reset fault is detected by $u$, then $u$ becomes a size-1 cluster in one round.  If no reset fault is detected, either (i) $u$ believes it is participating in a merge, or (ii) $u$ believes it is a member of a cluster.  For case (i), we will show later that the merge process is locally checkable (that is, if a configuration is reached that is not a valid merge, at least one node detects this), and that every node that detects an invalid cluster from a merge will either complete the merge in $\mathcal{O}(\log N)$ rounds, or reset in $\mathcal{O}(\log N)$ rounds, satisfying our claim.  For case (ii), as clusters are locally checkable, there must be a shortest path of nodes $u,v_0,v_1,\ldots,v_k$, with $k = \mathcal{O}(\log N)$, such that all nodes in the path have a cluster identifer matching the identifier of $u$, and where $v_k$ is the only node which detects a reset fault.  When $v_k$ executes a reset, it will cause node $v_{k-1}$ to detect a reset fault in the next round.  In this way, the reset will ``spread'' to $u$ in $\mathcal{O}(\log n)$ rounds, resulting in $u$ executing a reset, satisfying our claim.

\begin{lemma}
After $\mathcal{O}(\log N)$ rounds, if a set of nodes $T \subseteq V$ forms a cluster, no node in $T$ will execute a reset action until an external fault further perturbs the system.
\end{lemma}
\proofsketch This proof follows from the observation that once $u$ is part of a cluster $T$, no action $u$ executes will cause it to leave cluster $T$ unless it is merging with another cluster $T'$, which will successfully complete and result in a new valid cluster $T''$, with $V(T') \subseteq V(T'')$.  The initial ``delay'' of $\mathcal{O}(\log N)$ rounds is for a special circumstance: the case where a node is part of a cluster but the $PFC$ mechanism is corrupted.  This can only happen in the initial configuration, and it is corrected (either with the $PFC$ mechanism or through resets) in $\mathcal{O}(\log N)$ rounds, meaning our claim still holds.

The analysis from this point forward shall assume the system is in a ``reset-free'' configuration consisting of valid clusters and merging clusters.

\subsection{Matching}
For a merge to occur quickly, we can add edges to increase ``bandwidth'', but we must be careful to limit the resulting degree increase.  Therefore, we allow a cluster $T$ to merge with at most one other cluster at any particular time.  We achieve this by calculating a matching.  We say that a cluster $T$ has been assigned a \emph{merge partner} $T'$ if and only if the roots of $T$ and $T'$ have been connected by the matching process described below.  We say that a cluster $T$ is \emph{matched} if it has been assigned a merge partner, and unmatched otherwise.

Our goal is to find a large matching on the \emph{cluster graph} $G_c$ induced by the clusters in configuration $G$, where a node $v_T$ in $G_c$ corresponds to a cluster $T$ in $G$, and an edge $(v_T, v_{T'})$ corresponds to an edge between at least one node $u \in T$ and node $u' \in T'$.  To find a matching, we use a randomized symmetry-breaking technique.  Note there are topologies where even a maximum matching consists of only a small number of nodes (e.g. a star topology has a maximum matching of a single pair).  If merges were limited to those found by a matching, then, convergence would be slow.  However, one can identify large matchings on the square of the graph $G_c^2$ (the graph resulting from connecting all nodes of distance at-most 2 in $G$).  Since we are in the overlay network model, a matching on $G_c^2$ can become a (distance-one) matching by adding a single edge between matched clusters.  Our matching algorithm, then, creates a matching on the square of the cluster graph of the network, denoted $G_c^2$.  We provide a sketch of the matching algorithm in Algorithm \ref{algo:match_sketch}, and a discussion below.

To create this matching, our algorithm uses two different \emph{roles}, selected by the cluster root: \emph{leaders} and \emph{followers}.  Leaders connect followers together to form a matching on $G_c^2$.  A cluster root chooses the cluster's role uniformly at random, with the exception of one special case, discussed below (clusters which are merging become leaders if they were ``followed'' during their merge).

Consider a follower cluster.  We define two types of followers: \emph{long followers} and \emph{short followers}.  Short followers will only search for a leader for a ``short'' amount of time ($4 \log N$ rounds), while long followers will search for a (slightly) ``longer'' time ($24 \log N$).  Long and short followers are used to ensure the scenario where a cluster and all of its neighbors are ``stuck'' searching for a leader sufficiently rare.  Each node $u \in T$ will check $N(u)$ for a node $v$ such that $v$ is in another cluster $T'$ and $v$ is a potential leader.  A \emph{potential leader} is a node which either (i) has the role of leader and is ``open'' (see below), or (ii) is merging, and thus ``available'' for followers.  A node $u \in T$ will (i) mark one potential leader as ``followed'' (if one such neighbor exists), (ii) receive at most two edges to potential leaders from its children, and (iii) forward at most one edge incident on a potential leader to $u$'s parent.  Eventually, at most two such edges reach the root of the cluster.  At this point, the root waits for the selected leader to assign it a merge partner.

For the case where a root has selected the role of leader, the root begins by propagating the role to all nodes in the cluster.  At this point, nodes are considered \emph{open leaders}, and neighboring follower nodes can ``follow'' these leaders.  After this $PFC$ wave completes, the root sends another $PFC$ wave asking nodes in $T$ to (i) become \emph{closed leaders} (no node can select them as a potential leader), and (ii) connect any current followers as merge partners.  Nodes in $T$ will connect all followers incident upon them as merge partners, thus creating the matching on $G_c^2$.  If a node $u \in T$ has an odd number of followers, it simply matches the pairs of clusters, and then forwards the one ``extra'' edge to $u$'s parent.  This guarantees all followers will find a merge partner: the root of $T$ will either match the final two received followers, or set a follower as the merge partner for $T$.  Once the $PFC$ wave completes, the root either (i) begins the merge process with a follower $T'$ (if a merge partner was found), or (ii) randomly selects a new role.

\begin{algorithm}
\begin{tabbing}
........\=....\=....\=....\=....\=....\=....\kill
1.\>If no role, root $r_T$ selects a role uniformly at random: leader or follower.\\
2.\>If $r_T$ is a \emph{leader}:\\
3.\>\>$r_T$ uses $PFC$ to set all nodes as \emph{open leaders}\\
4.\>\>Upon completion of the wave, $r_T$ uses $PFC$ set all nodes as \emph{closed leader}\\
5.\>\>Upon completion of the wave, nodes connect all incident followers, and\\
\>\>forward to their parent the (at-most-one) unmatched follower\\
6.\>\>$r_T$ matches any received followers\\
7.\>\>$r_T$ either repeats the matching algorithm (if unmatched), or \\
\>\>\>begins merging (if matched)\\
8.\>Else\\
9.\>\>$r_T$ selects uniformly at random the role of \emph{long} or \emph{short} follower\\
10.\>\>Nodes in $T$ search for a leader.  Short followers search for $2$ $PFC$ waves,\\
\>\>while long followers search for $12$ $PFC$ waves\\
11.\>\>If a leader was found $r_T$ waits to be matched with a merge partner\\
12.\>\>Else $T$ repeats the matching algorithm\\
13.\>Endif
\end{tabbing}
\caption{The Matching Algorithm for Cluster $T$}
\label{algo:match_sketch}
\end{algorithm}

\subsubsection{Analysis of Matching}
The high-level ``idea'' behind our analysis is sketched below.

\begin{lemma}
Consider a cluster $T$ in a configuration $G$.  With probability at least $1/4$, all nodes in $T$ will be a potential leader for at least one round in the next $\mathcal{O}(\log N)$ rounds.
\end{lemma}
\proofsketch There are four cases to consider based upon the configuration of $T$: $T$ is currently an open leader, $T$ is currently a closed leader, $T$ is a follower, and $T$ is currently merging.  If $T$ is an open leader, our claim holds.  If $T$ is a closed leader, in $\mathcal{O}(\log N)$ rounds $T$ will either begin a merge or select the new role of leader with probability $1/2$.  If $T$ is a follower, $T$ is a short follower with probability $1/2$, and after $4(\log N + 1) + 4$ rounds $T$ will select a new role of leader with probability $1/2$, or begin a merge and become a potential leader.  If $T$ is currently merging, then every node will be a potential leader for at least one round during the merging process, which will complete in $\mathcal{O}(\log n)$ rounds.

\begin{lemma}
Consider a cluster $T$ in configuration $G$.  With probability at least $1/16$, $T$ is assigned at least one merge partner over $\mathcal{O}(\log N)$ rounds.
\end{lemma}
\proofsketch This proof combines the previous lemma with the fact that a cluster has probability $1/4$ of being a long follower, which will ensure the cluster searches sufficiently long to detect at least one potential leader in a neighboring cluster.

\subsection{Merging}
\label{section:merge}
Our merging algorithm adds edges in a systematic fashion so that (i) there is enough ``bandwidth'' for two clusters to merge quickly while (ii) limiting degree increases.  Our single merging algorithm can be discussed from two points of view: one as two $N$-node clusters in the guest network merging into a single $N$-node cluster, and another as two clusters of host nodes systematically updating their cluster successors and predecessors.  Below, we present a discussion from both viewpoints for clarity.  Note these are simply different ways of thinking about the same algorithm.

From the point of view of the guest network, merging can be thought of as (i) connecting guest nodes with identical identifiers from the two clusters, beginning with the roots, (ii) determining which of these guest nodes will remain in the new network (using the ranges of their hosts), and (iii) transferring the links from the ``deleted'' node to the ``winning'' node (the node remaining in the single merged $N$-node cluster).  The remaining node can then connect its children, and they can repeat the ``merge'' process.  This proceeds level-by-level until only a single $N$-node cluster remains.

From the point of view of the host network, the merge involves (i) connecting the two hosts of guest nodes with the same identifiers, and then (ii) updating the cluster ranges of these hosts, transferring any links from the ``lost range'' of one host node to another.  To see this, note that initially the hosts of the roots are connected, and clearly the cluster ranges of the two root hosts overlap.  The two hosts will update their cluster successor or predecessor (if needed), and the host whose cluster range was reduced will transfer any outgoing intra-cluster links to its new successor/predecessor, which has ``taken over'' that range.  This change in the cluster range corresponds to the ``deleting'' of a guest node discussed above.  Finally, the hosts of the children of the root are connected, and the merge can begin on level 1 of the tree.  This repeats recursively until all cluster successors, predecessors, and (by implication) Type 2 edges are updated in the host network, and a new cluster is formed.  The merge algorithm is sketched in Algorithm \ref{algo:merge_sketch}.

\begin{algorithm}
\begin{tabbing}
........\=....\=....\=....\=....\=....\=....\kill
\>// \emph{Cluster $T$ has been assigned merge partner $T'$}\\
1.\>Root $r_T$ notifies all nodes of merge partner $T'$ and\\
\>\>its view of the random sequence $\Psi_r$\\
2.\>Edges between $T$ and $T'$ are removed if $\Psi_r = \Psi$\\
3.\>Beginning with the roots $r_T$ and $r_{T'}$:\\
4.\>\>Node $r_T$ updates its range based upon the identifier of $r_{T'}$ (if needed)\\
5.\>\>Node $r_T$ sends any edges not in its new range to $r_{T'}$,\\
\>\>and receives edges from $r_{T'}$\\
6.\>\>Children of $r_T$ and $r_T'$ are connected, and process repeats concurrently\\
7.\>Once process reaches leaves, pass feedback wave to new root\\
8.\>New root $r_{T''}$ sends $PFC$ wave to update nodes in $T''$ of new cluster identifier.
\end{tabbing}
\caption{The Merging Algorithm for Cluster $T$}
\label{algo:merge_sketch}
\end{algorithm}

As a final note, every merge begins with a pre-processing stage which removes all links between merge partners $T$ and $T'$ besides the edge between roots.  To prevent this from disconnecting the network, no edge is deleted unless both incident nodes receive from their respective cluster roots a message matching the shared random sequence $\Psi$.  As this sequence is unknown to the adversary, we can prevent, with high probability, the network from being disconnected in the self-stabilizing setting.
\subsubsection{Analysis of Merge}
\begin{lemma}
Consider two clusters $T$ and $T'$ that are merge partners.  In $\mathcal{O}(\log N)$ rounds, $T$ and $T'$ have formed a single cluster $T''$ consisting of all nodes in $T \cup T'$.
\end{lemma}
\proofsketch The proof here follows from the fact that the merge process requires $\mathcal{O}(\log N)$ rounds of pre-processing, and then it resolves at least one level of the guest network in a constant number of rounds.  Since there are $\log N + 1$ levels, our lemma holds.

\begin{lemma}
The degree of a node $u \in T$ will increase by $\mathcal{O}(\log^2 N)$ during a merge, and will return to within $\mathcal{O}(\log N)$ of its initial degree when the merge is complete.
\label{lemma:merge_degree}
\end{lemma}
\proofsketch This follows from Theorem \ref{theorem:max_degree}.  For any set of nodes that forms a correct \Avatar (including a cluster), a node $u$ has at most $\mathcal{O}(\log N)$ edges in the host network.  As merging involves transferring all edges from a contiguous portion of $range(u)$ to some node $v$, and this occurs at most once per level of the guest network, no node will receive more than $\mathcal{O}(\log^2 N)$ edges.  Once the merge is completed, any node in the new cluster $T''$ has at most $\mathcal{O}(\log N)$ edges in $T''$, again by Theorem \ref{theorem:max_degree}.

\subsection{Termination Detection}
In our algorithm, the root of a cluster repeatedly executes the matching algorithm.  To ensure silent stabilization, we must inform the root when a legal configuration has been reached.  For this, we add a ``faulty bit'' to the feedback wave sent after a merge is complete.  If a node (i) detects the configuration is faulty, or (ii) received a faulty bit of $1$ from its children, the node sets its faulty bit to $1$ and appends it to its feedback message.  When the root receives a feedback wave without the faulty bit set (i.e. a value of $0$), it stops executing our algorithm.  A node $u$, upon completing this wave, can detect a ``reset fault'' whenever it finds its faulty bit is $0$ and it either (i) detects a faulty configuration, or (ii) detects a neighbor with a reset bit not equal to 0.  This ensures our algorithm is silent while remaining locally checkable.

\begin{lemma}
When our algorithm builds a legal \Avatar network, the faulty bit will be set to $0$, and remain $0$ until a transient fault again perturbs the system.
\end{lemma}
\proofsketch Since \Avatar is locally checkable, a faulty configuration has at least one detector which will set its faulty bit to $1$.  By similar argument to Lemma \ref{lemma:reset}, in $\mathcal{O}(\log N)$ rounds, all nodes will have their faulty bit set to $1$ and begin executing our algorithm.  Once the last merge occurs, no node will detect a fault, and all faulty bits will remain $0$ until another fault occurs.

\subsection{Combined Analysis}

\begin{theorem}
The algorithm in Section \ref{section:avatar_algo} is a self-stabilizing algorithm for the \Avatar network with expected convergence time of $\mathcal{O}(\log^2 N)$.
\end{theorem}
\proofsketch All nodes are members of a cluster in $\mathcal{O}(\log N)$ rounds, at which point the number of clusters will only decrease.  Each time a merge occurs, the number of clusters is reduced by 1, and the probability that a cluster merges over a span of $\mathcal{O}(\log N)$ rounds is constant ($1/16$).  In expectation, then, every cluster has merged in $\mathcal{O}(\log N)$ rounds, halving the number of clusters.  After $\mathcal{O}(\log^2 N)$ rounds, we are left with a single cluster, which is the legal configuration.

\Avatar also converges quickly in terms of space, which we show next.

\begin{theorem}
The degree expansion of the self-stabilizing $\textsc{Avatar}_{\textsc{Cbt}}$ algorithm from Section \ref{section:avatar_algo} is $\mathcal{O}(\log^2 N)$ in expectation.
\label{theorem:avatar_degree}
\end{theorem}
\proofsketch A node's degree will increase under only a small number of circumstances.  A node will only have $\mathcal{O}(\log N)$ edges to nodes in its cluster (except during merges).  During a merge, the degree can increase to at most $\mathcal{O}(\log^2 N)$ (Lemma \ref{lemma:merge_degree}).  Each time a cluster selects the leader role, a node in the cluster may have its degree increase by $1$.  As the algorithm will terminate in $\mathcal{O}(\log^2 N)$ rounds in expectation, there are an expected $\mathcal{O}(\log N)$ such increases.  Finally, consider an invalid initial configuration which causes a node $u$ to receive many edges while not in a cluster or merging with another cluster.  The only way for a node $u$ to receive additional edges in a round and not detect a fault not caused by merge (and thus execute a reset action) is for only a single node to send $u$ edges.  Furthermore, no node will send more than $\mathcal{O}(\log N)$ edges in a single round, or else a reset fault would be detected.  Since $u$ will be in a cluster in $\mathcal{O}(\log N)$ rounds, and each round might increase a node's degree by $\mathcal{O}(\log N)$ before this, we have shown the degree expansion of our algorithm is $\mathcal{O}(\log^2 N)$.

\section{Discussion and Future Work}
As a final topology, \Avatar suffers from poor load balancing, as it is built from a binary tree.  It can, however, be useful in creating other topologies.  We propose a mechanism we call \emph{network scaffolding} in which $\textsc{Avatar}_{\textsc{Cbt}}$ is used as an intermediate topology for stabilization from which another network is built on (much like a scaffold is used for construction).  Our technique has already been successful in building a self-stabilizing \textsc{Chord} network with polylogarithmic running time and degree expansion~\cite{berns_diss_12}.

Furthermore, we would like to relax the requirement that all nodes know $N$, perhaps using another self-stabilizing protocol.  We would also like to investigate bounds for the degree expansion to help determine if our algorithm (or any algorithm) is as efficient as can be expected in this setting.  We also are examining how much state nodes must exchange to guarantee local checkability (with mutable state, unlike proof labels~\cite{korman_prooflabel05}), exploring more-efficient ``stable'' configurations.

\noindent\emph{Acknowledgments:} I would like to thank my advisors, Dr. Sriram V. Pemmaraju and Dr. Sukumar Ghosh, for their guidance and discussions on this paper.

\bibliographystyle{splncs03}
\bibliography{sson}
%
%
%
%
%
%
%
%
%
%
%
%
%
%
%
%
%
%
%
%
%
%
%
%
%
%
%
%
%
%
%
%
%
%
%
%
%
%
%
%
%
%
%
%
%
%
%
%
%
%
%
%
%
%
%
%
%
%
%
%
%
%
%
%
%
%
%
%
%
%
%
%
%
%
%
%
%
%
%
%
%
%
%
%
%
%
%
%
%
%
\appendix
\appendixpage
\let\cleardoublepage\clearpage
\section{Avatar: Additional Details}
Below we give the full proof that $\textsc{Avatar}_{\textsc{Cbt}}$ has a logarithmic degree, regardless of the identifiers in $V$.

\begin{figure}
\centering
\includegraphics[scale=0.5]{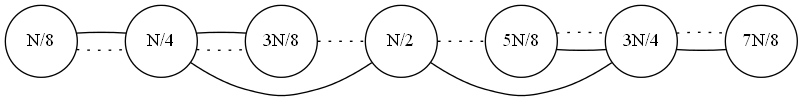}
\caption{A linear representation of levels 0-2 of a complete binary search tree.}
\label{dfig:cbt_edges}
\end{figure}

\begin{theorem}
For any node set $V \subseteq [N]$, the maximum degree of any node $u_i \in V$ in $\textsc{Avatar}_{\textsc{Cbt}}(N, V)$ is at most $2 \cdot \log N$.
\label{dtheorem:max_degree}
\end{theorem}
\begin{proof}
Consider $\Phi^{-1}(u_i)$, the subset of nodes from $[N]$ mapped to node $u_i$.  Let $[N]_j$ be the set of all nodes at level $j$ of $\textsc{Cbt}(N)$.  We show that, for any $0 \leq j \leq \log N + 1$, there are at most 2 nodes in $\Phi^{-1}(u_i) \cap [N]_j$ with a neighbor not in $\Phi^{-1}(u_i)$.

First, consider the edges in $\textsc{Cbt}(N)$.  Let the \emph{span of an edge $(a,b) \in E[\textsc{Cbt}(N)]$} ($a < b$) be all nodes from $[N]$ in the interval $(a,b)$, and let the \emph{size of a span} be $b - a$.  
Let a \emph{node segment} $S[a,b]$ be the contiguous set of nodes in the interval $[a,b]$ (that is, $S = \{c : a \leq c \leq b\}$).  Notice that, by definition of $\textsc{Cbt}(N)$, the spans of any two edges with the same span size are disjoint.  Therefore, for any segment $S$, there can be at most two edges in $E(\textsc{Cbt}(N))$ with the same span size going from a node $b \in S$ to a node $b' \notin S$.  Since there are $\log N$ span sizes in $\textsc{Cbt}(N)$, for any segment $S[a,b]$, there are at most $2 \cdot \log N$ edges with one endpoint in $S$ and one endpoint outside of $S$.

For example, consider Figure \ref{dfig:cbt_edges}, which is a linear representation of the first two levels of $\textsc{Cbt}(N)$, where $N$ is a power of 2.  Notice there are exactly 2 edges with spans of size $N/4$ (from $N/4$ to $N/2$ and $N/2$ to $3N/4$), and the span of the edges is disjoint.  Similarly, there are 4 edges that span $N/8$ points ($(N/8, N/4)$, $(N/4, 3N/8)$, $(5N/8, 3N/4)$, $(3N/4, 7N/8)$), also with disjoint spans.

Notice in the $\textsc{Avatar}_{\textsc{Cbt}}$ network, $\Phi^{-1}(u_i)$ is a node segment from $[N]$.  The only edges from nodes in $\Phi^{-1}(u_i)$ that require a Type 2 edge in $\textsc{Avatar}_{\textsc{Cbt}}(N,V)$ are those from some $b \in \Phi^{-1}(u_i)$ to $b' \notin \Phi^{-1}(u_i)$.  As there are most $2 \cdot \log N$ such edges, the degree of any node $u_i \in \textsc{Avatar}_{\textsc{Cbt}}(N,V)$ is at most $2 \cdot \log N + 2$ (at most two edges of type (1)).
\end{proof}

\section{Clustering: Additional Details}
\subsection{Full Algorithms}
\subsubsection{Reset}
Algorithm \ref{algo:reset} presents the simple reset algorithm nodes execute to ensure membership in a cluster.

\begin{algorithm}
\begin{tabbing}
........\=....\=....\=....\=....\=....\=....\kill
1.\>\textbf{if} a \emph{reset fault} is detected \textbf{and}\\
\>\>did not reset in previous round \textbf{then}\\
2.\>\>$\mathit{clusterSucc}_u = \bot$; $\mathit{clusterPred}_u = \bot$; $\mathit{cluster}_u = \mathit{id}_u$;\\
3.\>\textbf{fi}
\end{tabbing}
\caption{The Reset Algorithm}
\label{algo:reset}
\end{algorithm}

\subsubsection{PFC}
The modified $PFC$ algorithm as executed on the guest \textsc{Cbt} network is given in Algorithm \ref{algo:pfc}.  As mentioned in the original text, we assume when host nodes send a message using $PFC$, they append the sender's level from the guest network, allowing this to be simulated on the host network.  We use $I$ to represent the information being passed down, and $F$ to represent the information that should be fed back.  We also introduce the idea of a \emph{feedback action}, which allows us to specify node behavior that should be executed before completing the feedback portion of the $PFC$ wave.  In particular, we will use this action for a leader cluster to ``wait'' for followers to be ready to receive a merge partner, as well as for maintaining the faulty bit after a merge.

\begin{algorithm}
\begin{tabbing}
........\=....\=....\=....\=....\=....\=....\kill
\textbf{Variable for Node $a$:} $\mathit{PFCState}_a$\\
1.\>\textbf{when} node $\mathit{root}$ satisfies $\mathit{PFCState}_{\mathit{Children}(\mathit{root})} = \mathit{PFCState}_{\mathit{root}} = \mathit{Clean}$ \textbf{then}\\
2.\>\>$\mathit{root}_T$ initiates $\mathit{PFC}$ wave by setting $\mathit{PFCState}_{\mathit{root}} = \mathit{Propagate}(I)$\\
3.\>\>Each node $a$ executes the following:\\
4.\>\>\>\textbf{if} $\mathit{PFCState}_a = \mathit{Clean} \wedge \mathit{PFCState}_{\mathit{Parent}(a)} = \mathit{Propagate}(I) \wedge$\\
\>\>\>\>$\mathit{PFCState}_{\mathit{Children}(a)} = \mathit{Clean}$ \textbf{then}\\
5.\>\>\>\>\>$\mathit{PFCState}_a = \mathit{Propagate}(I)$\\
6.\>\>\>\textbf{else if} $\mathit{PFCState}_a = \mathit{Propagate}(I) \wedge \mathit{PFCState}_{\mathit{Parent}(a)} = \mathit{Propagate}(I) \wedge$\\
\>\>\>\>$\mathit{PFCState}_{\mathit{Children}(a)} = \mathit{Feedback}(F)$ \textbf{then}\\
7.\>\>\>\>\>$\mathit{PFCState}_a = \mathit{Feedback}(F')$\\
8.\>\>\>\textbf{else if} $\mathit{PFCState}_a = \mathit{Feedback}(F) \wedge \mathit{PFCState}_{\mathit{Parent}(a)} = \mathit{Feedback}(F) \wedge$\\
\>\>\>\>$\mathit{PFCState}_{\mathit{Children}(a)} = \mathit{Clean}$ \textbf{then}\\
9.\>\>\>\>\>$\mathit{PFCState}_a = \mathit{Clean}$\\
10.\>\>\>\textbf{fi}\\
11.\>\textbf{fi}
\end{tabbing}
\caption{Subroutine for $\mathit{PFC}(I,F)$ for guest network \textsc{Cbt}}
\label{algo:pfc}
\end{algorithm}

We can use the full definition of the $PFC$ mechanism to define the following particular type of cluster.

\begin{definition}
A set of nodes $T$ is called a \emph{proper cluster} when, for each $u \in T$, (i) the neighborhood of $u$ induced by nodes in $T$ matches $u$'s neighborhood in $\textsc{Avatar}_{\textsc{Cbt}}(N,T)$, (ii) cluster successors and predecessors of $u$ are consistent with $u$'s successor and predecessor in the graph induced by nodes in $T$, (iii) all nodes in $T$ have the same correct cluster identifier, (iv) no node in $T$ neighbors a node $v \notin T$ such that $\mathit{cluster}_v = \mathit{cluster}_u$, and (v) the communication mechanism is not faulty.
\end{definition}

\subsection{Clustering Analysis}
Given the definitions and algorithms above, we provide the full analysis of the clustering portion of our algorithm.

\begin{lemma}
If the root of a proper cluster $T$ initiates the $\mathit{PFC}(I, F)$ wave, the $\mathit{PFC}$ wave is complete (all nodes receive the information, the root receives the feedback, and nodes are ready for another $\mathit{PFC}$ wave) in $2 \cdot (\log N + 1) + 2$ rounds.
\label{dlemma:pfc}
\end{lemma}
\begin{proof}
After the root initiates the $\mathit{PFC}(I,F)$ wave, in every round the information $I$ moves from level $k$ to $k+1$ until reaching a leaf.  As the cluster has at most $\log N + 1$ levels, after at most $\log N + 1$ rounds, all nodes have received the propagation wave and associated information.  Upon receiving the propagation wave, leaves set their $\mathit{PFC}$ states to $\mathit{Feedback}$ to begin the feedback wave.  Again, in every round the feedback wave moves one level closer to the root, yielding at most $\log N + 1$ rounds before the root has received the feedback wave.  Consider the transition to $\mathit{PFC}$ state $\mathit{Clean}$.  When a leaf node $b$ sets its $\mathit{PFC}$ state to $\mathit{Feedback}$, in one round the parent of $b$ will set its state to $\mathit{Feedback}$, and in the second round, $b$ will set its state to $\mathit{Clean}$.  The process then repeats for the parent of $b$.  In general, two rounds after a node $b$ is in state $\mathit{Feedback}$, it transitions to state $\mathit{Clean}$.  Therefore, $2 \cdot (\log N + 1) + 2$ rounds after the root initiates a $\mathit{PFC}$ wave, all nodes receive the propagation wave, return the feedback wave, and move back to a clean state, ready for another $\mathit{PFC}$ wave.
\end{proof}

\begin{lemma}
Let node $b$ be a member of a proper cluster $T$.  Node $b$ can only execute a reset action if $T$ begins the merging process from Algorithm \ref{algo:merge}.
\label{dlemma:merge_reset}
\end{lemma}
\begin{proof}
To begin, notice that the cluster structure is not changed unless a node executes a merge action.  Furthermore, the $\mathit{PFC}$ algorithm will never cause a reset fault in a proper cluster.  Therefore, no reset is executed unless a merge action is executed.
\end{proof}

\begin{lemma}
Let a guest node $r_T$ be a node whose immediate neighborhood matches the neighborhood of a correct root node ($r_T$ has a consistent $\mathit{PFC}$ state and correct cluster neighborhood).  If $r_T$ initiates a $\mathit{PFC}(I, F)$ wave and later receives the corresponding feedback wave, then $\mathit{r}_T$ is the root of a proper cluster.
\label{dlemma:pfc_works}
\end{lemma}
\begin{proof}
Every node will only continue to forward the propagate and feedback waves if (i) they have the appropriate cluster neighbors, and (ii) their $\mathit{PFC}$ states are consistent.  Therefore, if $r_T$ receives a feedback wave, $T$ must be a proper cluster.
\end{proof}

We classify clusters based upon their $\mathit{PFC}$ state next.

\begin{definition}
Let $T$ be a proper cluster.  $T$ is a \emph{proper clean cluster} if and only if all nodes in $T$ have a $\mathit{PFC}$ state of $\mathit{Clean}$.
\end{definition}

Let $\mathcal{F}(G_i)$ be all configurations reached by executing our algorithm beginning in configuration $G_i$ (that is, the set of future configurations).  We now bound the  occurrences of reset after a given configuration.

\begin{lemma}
If node $b$ is a member of a proper unmatched clean cluster $T$ in configuration $G_i$, then $b$ will never execute a reset in any configuration $G_j \in \mathcal{F}(G_i)$.
\label{dlemma:clean_merge}
\end{lemma}
\begin{proof}
By Lemma \ref{dlemma:merge_reset}, only a merge can cause node $b$ to execute a reset fault.  We show that any merge $b$ participates in must be between two proper clusters, and therefore completes correctly (Lemma \ref{dlemma:merge}).

Suppose the root of $T$ has been matched with the root of another cluster $T'$.  $T$ cannot begin modifying its cluster edges for the merge until both $T$ and $T'$ have successfully completed the $\mathit{PFC}(Prep(T,T'),\bot)$ wave.  Suppose $T'$ was not a proper cluster.  In this case, $T'$ would not successfully complete the $\mathit{PFC}$ wave (Lemma \ref{dlemma:pfc_works}), and $T$ would not begin a merge with $T'$.  Therefore, if $T$ and $T'$ merge together, both must be proper clusters, implying the merge completes successfully and $b$ is again a member of a clean proper cluster $T''$ (see Lemma \ref{dlemma:merge}).
\end{proof}

\begin{lemma}
Consider a node $b$ that is not a member of a proper cluster in configuration $G_i$.  In $\mathcal{O}(\log N)$ rounds, $b$ is a member of a proper unmatched clean cluster.
\label{dlemma:reset_to_clean}
\end{lemma}
\begin{proof}
If $b$ is not a member of a proper cluster and detects a reset fault in $G_i$, our lemma holds.

Consider the case where $b$ is not a member of a proper cluster but has no reset fault.  If $b$'s cluster neighborhood is not a legal cluster neighborhood, then $b$ must be performing a merge between its own cluster $T$ and neighboring cluster $T'$.  In 2 rounds, $b$ either has the correct cluster neighbors, each with tree identifier of either $T$ or $T'$, and has passed the merge process on to its children, or $b$ has executed a reset.  The children of $b$ now either execute a reset, or are in a state consistent with the merge process, and we repeat the argument.  As there are $\log N + 1$ levels, after $2 \cdot (\log N + 1)$ rounds either a node has reset due to this merging, or the merge is complete.  If a node has reset in round $i$, its parent will reset in round $i+1$, its parent's parent will reset in $i+2$, and so on.  After at most $\log N + 1$ rounds, $b$ resets and becomes part of an unmatched clean proper cluster.

If $b$ does not detect locally that it is not a member of proper cluster $T$, then there must exist a node $c$ within distance $2 \cdot \log N$ such that every node $p_i$ on a path from $b$ to $c$ believes it is part of the same cluster as $b$, and node $c$ detects an incorrect cluster neighborhood or inconsistent $\mathit{PFC}$ state.  As with the above argument, either $c$ detects a reset fault immediately, or $c$ is participating in a merge.  In either case, after $\mathcal{O}(\log N)$ rounds, $c$ has either reset, or $c$ is a member of a proper unmatched clean cluster resulting from a successful merge.
\end{proof}

Combining Lemmas \ref{dlemma:clean_merge} and \ref{dlemma:reset_to_clean} gives us the following lemma.

\begin{lemma}
No node executes a reset action after $\mathcal{O}(\log N)$ rounds.
\end{lemma}

We call a configuration $G_i$ a \emph{reset-free configuration} if and only if no reset actions are executed in any configuration $G_j \in \mathcal{F}(G_i)$.  For the remainder of our proofs in this appendix, we shall assume a reset-free configuration.

\section{Matching: Additional Details}
\subsection{Full Algorithms}
Below we include the full algorithm for the matching process done by a leader cluster (Figure \ref{algo:lead}), the subroutine used by leaders to create the matching among followers (Figure \ref{algo:connect}), and the algorithm followed by a follower cluster (Figure \ref{dalgo:follow}).

\begin{algorithm}
\begin{tabbing}
........\=....\=....\=....\=....\=....\=....\kill
\>\textit{// The root $r_T$ of $T$ has selected the leader role}\\
1.\>Node $r_T$ uses $PFC$ to inform all nodes in $T$ of\\
\>\>leader role\\
2.\>Node $r_T$ uses $PFC$ to close all nodes in $T$\\
3.\>Node $r_T$ initiates the $\mathit{ConnectFollowers}$ procedure (Algorithm \ref{algo:connect})\\
4.\>\textbf{if} $T$ was not assigned a merge partner \\
\>\>\>during $\mathit{ConnectFollowers}$ \textbf{then}\\
5.\>\>$r_T$ randomly selects a new role\\
6.\>\textbf{fi}
\end{tabbing}
\caption{The Matching Algorithm for a Leader Cluster}
\label{algo:lead}
\end{algorithm}

\begin{algorithm}
\begin{tabbing}
........\=....\=....\=....\=....\=....\=....\kill
1.\>Execute $\mathit{PFC}(\mathit{ConnectFollowers}, \bot)$:\\
2.\>\>\textbf{Feedback Action for $a$:}\\
3.\>\>\>\textbf{while} $\exists b \in N_a : \mathit{role}_b = \mathit{PotentialFollower}(a) \vee $\\
\>\>\>\>$(\mathit{role}_b = \mathit{Follower}(a) \wedge b \neq \mathit{root})$ \textbf{do}\\
4.\>\>\>\>skip;\\
5.\>\>\>\textbf{od}\\
4.\>\>\>Order the $k$ followers in $N_a$ by tree\\
\>\>\>\> identifiers $b_0, b_1, b_2, \ldots, b_{k-1}$\\
5.\>\>\>\textbf{for} $i = 0,2,4,\ldots,\lfloor k/2 \rfloor$ \textbf{do}\\
6.\>\>\>\>Create edge $(b_i, b_{i+1})$;\\
7.\>\>\>\>Set merge partner of $b_i$ to $b_{i+1}$ and vice versa\\
7.\>\>\>\>Delete edge $(a,b_{i+1})$\\
8.\>\>\>\textbf{od}\\
9.\>\>\>\textbf{if} $k \bmod 2 \neq 0$ \textbf{then}\\
10.\>\>\>\>Create edge $(\mathit{parent}_a,b_{k-1})$\\
11.\>\>\>\>Delete edge $(a,b_{k-1})$\\
11.\>\>\>\textbf{fi}
\end{tabbing}
\caption{Subroutine $\mathit{ConnectFollowers}$}
\label{algo:connect}
\end{algorithm}

\begin{algorithm}
\begin{tabbing}
........\=....\=....\=....\=....\=....\=....\kill
\>\textit{// Assume root $r_T$ of $T$ has selected follower role}\\
1.\>Root $r_T$ selects a role of \emph{short} or \emph{long} follower uniformly at random.\\
2.\>Root $r_T$ propagates role of follower to nodes in $T$.\\
3.\>\textbf{if} $T$ is a \emph{short follower} \textbf{then}\\
4.\>\>Root $r_T$ sets $\mathit{pollCnt} = 2$\\
5.\>\textbf{else} $T$ is a \emph{long follower}:\\
6.\>\>Root $r_T$ sets $\mathit{pollCount} = 12$\\
7.\>\textbf{fi}\\
8.\>\textbf{while} $\mathit{pollCnt} > 0$ \textbf{and} no potential leader found \textbf{do}\\
9.\>\>Root $r_T$ queries cluster $T$ (using $PFC$) for a potential leader.\\
10.\>\>\textbf{if} $u \in T$ found a potential leader $v \in T'$ \textbf{then}\\
11.\>\>\>Forward one edge $v$ to parent of $u$ during feedback wave\\
12.\>\>\textbf{fi}\\
13.\>\>Root $r_T$ sets $\mathit{pollCnt} = \mathit{pollCnt} - 1$;\\
14.\>\textbf{od}\\
15.\>\textbf{if} a potential leader is returned to $r_T$ \textbf{then}\\
16.\>\>Root node $r_T$ selects one potential leader $v \in T'$\\
\>\>\>and informs nodes in $T$ of $v$\\
17.\>\textbf{else} \\
18.\>\>$r_T$ randomly selects a new role from $(leader, follower)$\\
19.\>\textbf{fi}
\end{tabbing}
\caption{The Matching Algorithm for a Follower Cluster}
\label{dalgo:follow}
\end{algorithm}

\subsection{Analysis of Matching}
\begin{lemma}
Let $b \in T$ be a follower that has selected a neighbor $c \in T'$ as a potential leader.  In at most $5 \cdot (\log N + 1) + 6$ rounds, the root of $T$ has an edge to some leader cluster $T''$, and all nodes in $T$ know this leader.
\label{dlemma:follow_to_root}
\end{lemma}
\begin{proof}
When $b$ detects $c$ becomes a potential leader, $b$ marks $c$ as a potential leader immediately, regardless of the $\mathit{PFC}$ state.  After at most $2 \cdot (\log N + 1) + 2$ rounds, a feedback wave will reach $b$, at which point $b$ will forward the information about its potential leader.  In an additional $(\log N + 1) + 2$ rounds, the $\mathit{PFC}$ wave completes, at which point the root of $T$ has at least one potential leader (which may or may not be $c$) returned to it.  The root of $T$ will select one returned leader and execute the leader-inform $\mathit{PFC}$ wave.  In $\log N + 1$ rounds, all nodes know the identity of their leader and its cluster.  Cluster $T$'s selected leader $c'$ will be forwarded up the tree during the feedback wave, reaching the root in an additional $(\log N + 1) + 2$ rounds.
\end{proof}

\begin{lemma}
Let $r_T$ be the root of cluster $T$.  If $r_T$ selects the role of $\mathit{Leader}$, within $9 \cdot (\log N + 1) + 10$ rounds either $T$ has been paired with a merge partner, or $T$ randomly selects a new role.  Furthermore, all followers of $T$ have been assigned a merge partner.
\label{dlemma:lead_max}
\end{lemma}
\begin{proof}
First, note that $\mathit{PFC}(\mathit{Lead}, \bot)$ requires $2 \cdot (\log N + 1) + 2$ rounds to complete (Lemma \ref{dlemma:pfc}).  The $\mathit{PFC}(\mathit{ConnectFollowers}, \bot)$ wave requires at most $7 \cdot (\log N + 1) + 8$ rounds.  To see this, notice that the feedback action for this wave cannot advance past a node $b \in T$ until $b$ has no neighbors that are potential followers and all followers are root nodes.  Let $T'$ be a follower cluster that has selected $T$ as a potential leader.  By Lemma \ref{dlemma:follow_to_root}, after at most $5 \cdot (\log N + 1) + 6$ rounds, all potential followers of $b$ are either no longer following $b$, or are connected with a root to $b$.

Notice that the \emph{total} wait for all nodes in $T$ is $5 \cdot (\log N + 1) + 6$, since all nodes in $T$ have a role of $\mathit{ClosedLead}$ and are not assigned any additional potential followers.  Therefore, the feedback wave can be delayed at most $5 \cdot (\log N + 1) + 6$ rounds, leading to a total $7 \cdot (\log N + 1) + 8$ rounds for the $\mathit{PFC}(\mathit{ConnectFollowers}, \bot)$ wave.  When the $\mathit{PFC}(\mathit{ConnectFollowers},\bot)$ wave completes, all followers of $T$ have been assigned a merge partner.  If there were an odd number of followers, $T$ has also been assigned a merge partner, else $T$ will randomly re-select a role.
\end{proof}

\begin{lemma}
Let $T$ be a short follower cluster.  In at most $4 \cdot (\log N + 1) + 4$ rounds, either $T$ has selected a leader, or $T$ randomly re-selects a role.
\label{dlemma:short_follow}
\end{lemma}
\begin{proof}
A short follower polls its cluster for a leader at most twice, each requiring $2 \cdot (\log N + 1) + 2$ rounds (Lemma \ref{dlemma:pfc}).  If a leader is not returned, $T$ will randomly select a new role.  If at least one leader is returned, $T$ selects it.
\end{proof}

\begin{lemma}
Let $T$ be a long follower cluster.  In at most $24 \cdot (\log N + 1) + 24$ rounds, either $T$ has selected a leader, or $T$ randomly re-selects a role.
\label{dlemma:long_follow}
\end{lemma}
\begin{proof}
By similar argument to Lemma \ref{dlemma:short_follow}, a long follower polls its cluster at most 12 times, each requiring $2 \cdot (\log N + 1) + 2$ rounds (Lemma \ref{dlemma:pfc}).  If a leader is not returned, $T$ will randomly select a new role.  If at least one leader is returned, $T$ selects it.
\end{proof}

%
\begin{lemma}
Let $T$ be a follower cluster that has returned a leader after a $\mathit{PFC}$ search wave of Algorithm \ref{dalgo:follow}.  After at most $16 \cdot (\log N + 1) + 16$ rounds, $T$ has a merge partner.
\label{dlemma:follow_success_unmatched}
\end{lemma}
\begin{proof}
Let $r_T$ be the root node of $T$.  Node $\mathit{r}_T$ selects a returned leader $T'$ and, in $2 \cdot (\log N + 1) + 2$ additional rounds, all nodes in $T$ have been informed of leader $T'$ and $\mathit{r}_T$ has an edge to a node $b$ from $T'$.

After $\mathit{root}_T$ has an edge to a node $b \in T'$, $\mathit{root}_T$ waits to be assigned a merge partner.  Suppose $T'$ had the role of leader when selected by $T$.  By Lemma \ref{dlemma:lead_max}, after at most $9 \cdot (\log N + 1) + 10$ rounds, $T$ will be assigned a merge partner.

Suppose $T'$ was executing a merge when selected by $T$.  $T$ will be assigned a merge partner when (i) $T'$ finishes its merge, and (iii) $T'$ finishes executing Algorithm \ref{algo:lead}.  If $T'$ was merging, it completes in at most $5 \cdot (\log N + 1) + 4$ rounds (Lemma \ref{dlemma:merge}).  The resulting cluster $T''$ will be a leader, and will require at most $9 \cdot (\log N + 1) + 10$ rounds before all followers have been assigned a merge partner (Lemma \ref{dlemma:lead_max}).
\end{proof}

Since the initial configuration is set by the adversary, it can be difficult to make probabilistic claims when dealing with the initial configuration.  For instance, the adversary could assign all clusters the role of long follower, in which case the probability that a merge happens over $24 \cdot (\log N +1) + 24$ rounds is 0.  Notice, however, that after a short amount of time, regardless of the initial configuration, clusters are guaranteed to have randomly selected their current roles.  Therefore, we ignore the first $24 \cdot (\log N + 1) + 24$ rounds of execution in the following lemmas.

\begin{definition}
Let $G_0$ be the initial network configuration.  We define $\mathcal{F}_{\Delta}(G_0)$ to be the set of future configurations reached after $\Delta = 26 \cdot (\log N + 1) + 26$ rounds of program execution from the initial configuration.
\end{definition}

\begin{lemma}
Let $T$ be a follower cluster in configuration $G_i \in \mathcal{F}_{\Delta}(G_0)$.  With probability at least $1/2$, $T$ either randomly selects a new role or has found a leader in $4 \cdot (\log N + 1)$ rounds.
\label{dlemma:follow_to_lead}
\end{lemma}
\begin{proof}
Cluster $T$ must have randomly selected its follower role in $G_i$, as no cluster can be a follower for longer than $24 \cdot (\log N + 1) + 24$ rounds.  Given that $T$ is a follower, with probability $1/2$, $T$ must have been a short follower, and therefore either $T$ finds a leader or selects a new role after $4 \cdot (\log N + 1) + 4$ rounds (Lemma \ref{dlemma:short_follow}).
\end{proof}

\begin{lemma}
Consider configuration $G_i \in \mathcal{F}_{\Delta}(G_0)$.  With probability at least $1/4$, every node in cluster $T$ will have been a potential leader for at least one round over the next $21 \cdot (\log N + 1) + 20$ rounds.
\label{dlemma:lead_probability}
\end{lemma}
\begin{proof}
Consider the possible roles and states of any cluster $T$.  If $T$ is currently participating in a merge or is an $\mathit{OpenLeader}$, then our lemma holds.

Suppose $T$ is a follower in configuration $G_i$.  By Lemma \ref{dlemma:follow_to_lead}, with probability at least $1/2$, $T$ will either find a leader or randomly select another role after $4 \cdot (\log N + 1) + 4$ rounds.  If $T$ finds a leader, after an additional $16 \cdot (\log N + 1) + 16$ rounds (Lemma \ref{dlemma:follow_success_unmatched}), $T$ is assigned a merge partner, and after at most $\log N + 1$ additional rounds, all nodes in $T$ are potential leaders.  If $T$ randomly selects another role, with probability $1/2$ $T$ selects the leader role, and all nodes are potential leaders after at most $\log N + 1$ additional rounds.

Suppose a node $b \in T$ is a closed leader ($\mathit{role}_b = \mathit{ClosedLeader}$).  After at most $9 \cdot (\log N + 1) + 10$ rounds (Lemma \ref{dlemma:lead_max}), either the root of $T$ is assigned a merge partner and $b$ becomes a potential leader after an additional $\log N + 1$ rounds, \emph{or} the root of $T$ is not assigned a merge partner and selects a new role at random.  With probability $1/2$, then, $b$ becomes a potential leader after an additional $\log N + 1$ rounds.
\end{proof}

\begin{lemma}
Every cluster $T$ in configuration $G_i \in \mathcal{F}_{\Delta}(G_0)$ has probability at least $1/16$ of being assigned a merge partner over $64 \cdot (\log N + 1) + 64$ rounds.
\end{lemma}
\begin{proof}
After at most $24 \cdot (\log N + 1) + 24$ rounds, if $T$ has not been assigned a merge partner, $T$ will re-select its role.  With probability $1/4$, $T$ will be a long follower and be searching for a leader for $24 (\log N + 1) + 24$ rounds.  By Lemma \ref{dlemma:lead_probability}, a neighboring cluster $T'$ has probability at least $1/4$ of being a potential leader during this time.  Therefore, $T$ has probability at least $1/16$ of selecting a leader within $48 \cdot (\log N + 1) + 48$ rounds, which will result in $T$ being assigned a merge partner after at most an additional $16 \cdot (\log N + 1) + 16$ rounds (Lemma \ref{dlemma:follow_success_unmatched}).
\end{proof}

\section{Merge: Additional Details}
\subsection{Additional Figures}
To help see the merge process, Figure \ref{fig:merge1} contains two steps of a merge procedure, depicted both on the host network $\textsc{Avatar}_{\textsc{Cbt}}$ and on the guest network $\textsc{Cbt}$.

\begin{figure}
\centering
\subfigure[Step 0 of Merge (Guest Network)]{
   \label{fig:guest_merge0}
   \includegraphics[scale=0.23]{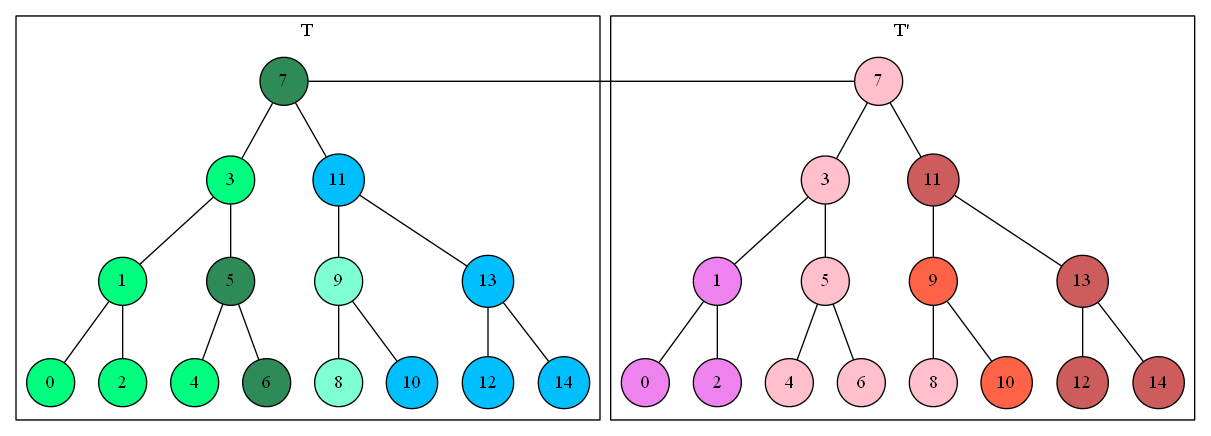}
}\hfill
\subfigure[Step 0 of Merge (Host Network)]{
   \label{fig:real_merge0}
   \includegraphics[scale=0.20]{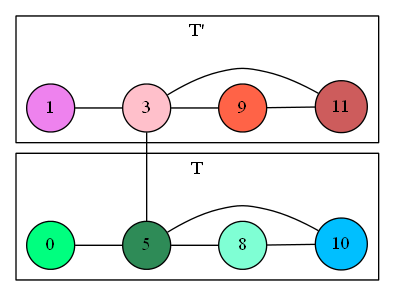}
}
%
\subfigure[Step 1 of Merge (Guest Network)]{
   \label{fig:virtual_merge1}
   \framebox{\includegraphics[scale=0.23]{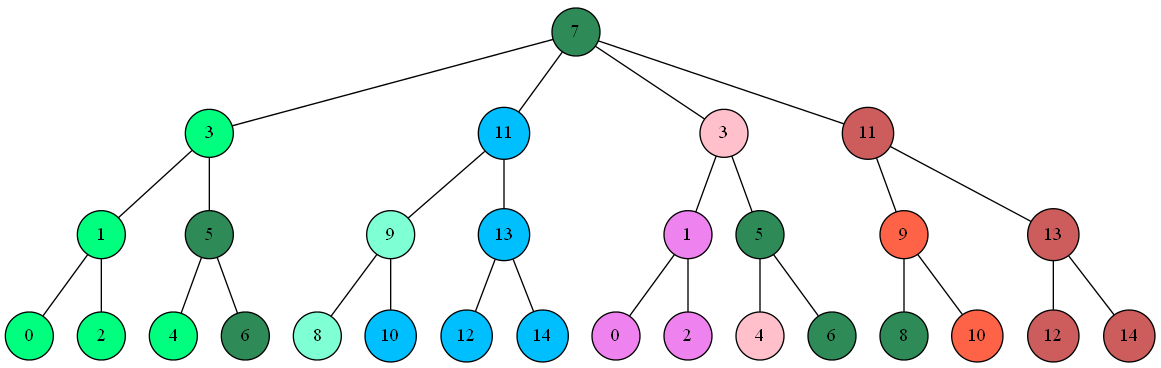}}
}\hfill
\subfigure[Step 1 of Merge (Host Network)]{
   \label{fig:real_merge1}
   \framebox{\includegraphics[scale=0.20]{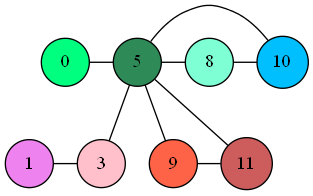}}
}
\caption{The guest and host networks for clusters $T$ and $T'$ at the start of a merge, and after the first step of a merge.  After the merge step, only one root exists in the guest network (the root hosted by $5$), and the successor of node $3$, along with 3's range, has been updated.}
\label{fig:merge1}
\end{figure}

\subsection{Full Algorithms}
We provide the full algorithms used for merging in our self-stabilizing algorithm: the procedure used in the guest network for replacing a guest node with another (Figure \ref{algo:replace_node}), and the algorithm for merging the entire cluster (Figure \ref{algo:merge}).

\begin{algorithm}
\begin{tabbing}
........\=....\=....\=....\=....\=....\=....\kill
$\mathit{ReplaceNode}(c, d):$\\
1.\>\textbf{if} $\mathit{partner}_a \neq \mathit{cluster}_b \vee \mathit{partner}_b \neq \mathit{cluster}_a$\\
\>\>$\vee \mathit{rs}_a \neq L \vee \mathit{rs}_b \neq L$ \textbf{then}\\
2.\>\>Reset hosts of nodes $c$ and $d$ (ends the merge process)\\
3.\>\textbf{fi}\\
4.\>\textbf{if} $\mathit{host}_d < \mathit{clusterSucc}_{\mathit{host}_c}$ \textbf{then}\\
5.\>\>$\mathit{clusterSucc}_{\mathit{host}_c} = \mathit{host}_d$\\
6.\>\>$\mathit{LostNodes}_c = \{b : \mathit{host}_b = \mathit{host}_c \wedge b > \mathit{clusterSucc}_{\mathit{host}_c}\}$\\
7.\>\textbf{else if} $\mathit{clusterPred}_{\mathit{host}_c} = \bot \wedge \mathit{host}_d < \mathit{host}_c$ \textbf{then}\\
8.\>\>$\mathit{clusterPred}_{\mathit{host}_c} = \mathit{host}_d$;\\
9.\>\>$\mathit{LostNodes}_c =$\\
\>\>\>$\{b : \mathit{host}_b = \mathit{host}_c \wedge \mathit{host}_d < b < \mathit{host}_c\}$\\
10.\>\textbf{else} \textit{// No successor pointer is updated}\\
11.\>\>Connect cluster children of $c$ to $d$; Delete node $c$\\
12.\>\textbf{fi}\\
13.\>\textbf{for each} $a \in \mathit{LostNodes}_c$ \textbf{do}\\
14.\>\>Copy node $a$ and tree neighbors to $\mathit{host}_d$\\
15.\>\>Delete $a$ from $\mathit{host}_c$'s embedding\\
16.\>\textbf{od}
\end{tabbing}
\caption{The $\mathit{ReplaceNode}$ Procedure}
\label{algo:replace_node}
\end{algorithm}

\begin{algorithm}
\begin{tabbing}
........\=....\=....\=....\=....\=....\=....\kill
\>\textbf{Precondition:} $T$ and $T'$ are merge partners with\\
\>\>connected roots.\\
1.\>$\mathit{root}_T$ ($\mathit{root}_{T'}$) notifies $T$ ($T')$ of \\
\>\>(i) merge partner $T'$ ($T$), and\\
\>\>(ii) value of the shared random sequence.\\
2.\>Remove all \emph{matched edges between $T$ and $T'$}.\\
3.\>$\mathit{ResolveCluster}(\mathit{root}_T, \mathit{root}_{T'})$\\
4.\>Once $\mathit{ResolveCluster}$ completes at leaves,\\
\>\>inform nodes in new cluster $T'' = T \cup T'$ about\\
\>\>new cluster identifier\\
\\
$\mathit{ResolveCluster}(a,b):$ for $a \in \textsc{Cbt}_T(N), b \in \textsc{Cbt}_{T'}(N)$\\
\>\>\textit{// without loss of generality, assume $a \prec b$}\\
1.\>\>$\mathit{ReplaceNode}(a,b)$\\
\>\>\textit{// Node $b$ is now connected to children of $a$.}\\
\>\>\textit{// Let $l_a$ ($r_a$) be the left (right) child of $a$,}\\
\>\>\textit{// and $l_b$ ($r_b$) be the left (right) child of $b$.}\\
2.\>\>Create edges $(l_a, l_b)$ and $(r_a, r_b)$;\\
3.\>\>Concurrently execute $\mathit{ResolveCluster}(l_a, l_b)$\\
\>\>\>and $\mathit{ResolveCluster}(r_a, r_b)$
\end{tabbing}
\caption{The Merge Algorithm}
\label{algo:merge}
\end{algorithm}

\subsection{Analysis of Merging}
We present the following lemma concerning the time required to complete a merge between two clusters $T$ and $T'$.

\begin{lemma}
Let $T$ and $T'$ be proper clusters, and let the merge partner of $T$ ($T'$) be $T'$ ($T$).  Assume the root $r_T$ of $T$ and the root $r_{T'}$ of $T'$ are connected.  In $5 \cdot (\log N + 1) + 4$ rounds, $T$ and $T'$ have merged together into a single proper unmatched clean cluster $T''$, containing exactly $N$ nodes.
\label{dlemma:merge}
\end{lemma}
\begin{proof}
The first step of the merge procedure is to execute the $\mathit{PFC}(\mathit{Prep})$ wave, which requires $2 \cdot (\log N + 1) + 2$ rounds (Lemma \ref{dlemma:pfc}).  Consider an invocation of the procedure $\mathit{ResolveCluster}(a,b)$.  Let $a$ be from cluster $T$, $b$ be from cluster $T'$, and without loss of generality let $a$ be the guest node which is to be deleted (that is, the range of the host of $b$ will contain $a$ after the merge).  $\mathit{ReplaceNode}(a,b)$ requires only 1 round, and results in the children of $a$ being connected to node $b$.  In the next round, $b$ will connect its children with the children from $a$, which requires 1 round.  $\mathit{ResolveCluster}$ is then executed concurrently for nodes from level $i+1$.  Therefore, the running time starting from level $i$ is $T(i) = 2 + T(i+1)$.  Since there are $\log N + 1$ levels, we have $T(0) = \sum_{i=0}^{\log N}{2} = 2 \cdot (\log N + 1)$.  After the resolution process reaches the leaves, the final feedback travels up the tree, requiring an additional $\log N + 1$ rounds, plus $2$ rounds for cleaning.
\end{proof}

\begin{lemma}
With probability at least $(1 - N/2^k)$ (where $k = |L|$ and $k \geq \log N$), the algorithm does not disconnect the network.
\end{lemma}
\begin{proof}
First, notice that deletions (of edges and nodes) that occur due to proper clusters $T$ and $T'$ merging do not disconnect the network -- the only edges deleted are those between nodes in $T \cup T'$, and these nodes will form a proper cluster $T''$ after the merge.  The only way in which the network can be disconnected is if the adversary creates an initial configuration such that a node $b$ believes it is either merging, or preparing for a merge, and thus deletes an edge to a node $c$.  Notice that for any edge $(b,c)$ to be deleted, both $b$ and $c$ must have the same value for their random sequence, and this value must match the shared random string $L$.  While the adversary can enforce the first condition, they are unable to guarantee the second.  Instead, for any particular pair of nodes $a$ and $b$, the adversary has probability $1/2^k$ of setting the random sequences of $a$ and $b$ to match $L$.  An adversary can have up to $N/2$ different ``guesses'' in any initial configuration.  Therefore, the probability that the network is disconnected is at most $N/2^{k+1}$ (for $k \geq \log N$).
\end{proof}

\subsection{Degree Expansion Analysis}
In this section, we describe in full the analysis used to show the degree expansion of our algorithm is $\mathcal{O}(\log^2 N)$.

To begin, we present the following corollary, which is a result of Theorem \ref{theorem:max_degree}.

\begin{corollary}
Consider a node $u$ hosting a set of nodes $\mathit{Virtual}_u$ such that all $b \in \mathit{Virtual}_u$ belong to the same proper cluster $T$.  Node $u$ has at most $2 \cdot \log N$ virtual nodes with neighbors in cluster $T$ that are not hosted by $u$.
\label{dcorollary:cluster_degree}
\end{corollary}

We define the set of actions a node may execute that can increase the degree of a real node $u$.

\begin{definition}
Let a \emph{degree-increasing action} of a virtual node $b$ be any action that adds a node $c$ to the neighborhood of a node $b' \in N_b$ such that $b'$ is not hosted by $\mathit{host}_b$.  Specifically, the degree-increasing actions are:
\begin{enumerate}
\item (Selection for Leaders): edges added from the connecting and forwarding of followers during the $\mathit{PFC}(\mathit{ConnectFollowers},\bot)$ wave of Algorithm \ref{algo:lead}
\item (Selection for Followers): forwarding an edge to a leader after the root has selected a leader in Algorithm \ref{dalgo:follow}
\item (Merge): resolution and virtual node transfer actions during Algorithm \ref{algo:merge}
\end{enumerate}
\end{definition}

Notice that transferring all non-cluster edges from a loser node is not a degree-increasing action, as the edges are ``virtual transfers'' between two virtual nodes hosted by the same real node.

\begin{lemma}
Let $u$ be a real node in configuration $G_i$.  The maximum number of real nodes $u$ will add to any neighbor $v$'s neighborhood in a single round is $2 \cdot \log N$.
\label{dlemma:single_node_increase}
\end{lemma}
\begin{proof}
We consider the degree-increasing actions.  Notice that a real node will detect a reset fault if it hosts two virtual nodes $b$ and $b'$ such that $b$ and $b'$ are executing different algorithms -- for example, if $b$ is merging while $b'$ is executing a selection for leaders, host $u$ will reset and not forward any neighbors.

Consider the selection algorithms for both leaders and followers as executed on a virtual node $b$.  Node $b$ can only increase the degree of its parent or of a follower.  Consider the degree increase $b$ causes to its parent.  Node $b$ may give its parent a single edge to a follower or a leader.  Since $u$ hosts at most $2 \cdot \log N$ virtual nodes with cluster neighbors not hosted by $u$, and each of these neighbors can increase the degree of a node by at most 1, node $u$ can only increase the degree of a cluster neighbor when executing selection for leaders and followers by at most $2 \log N$.

Next, consider how virtual node $b$ may increase the degree of a follower with Algorithm \ref{algo:lead}.  Every follower $b'$ of $b$ may have one additional edge added by $b$.  Notice, however, that every follower $b'$ must have a unique host -- if not, this host would detect a reset fault.  Therefore, $u$ can increase the degree of a real node $v$ by at most 1 when connecting followers in Algorithm \ref{algo:lead}.

Consider the merge actions of virtual nodes hosted by $u$.  Again, node $u$ must have all virtual nodes executing the merge algorithm, else $u$ resets.  In a given round, $u$ may update its successor and give up all hosted virtual nodes in a particular range to its new successor $v$.  The virtual nodes in this range can have at most $2 \cdot \log N$ real neighbors.
\end{proof}

\begin{lemma}
Let $u$ be a real node in some configuration $G_i$.  The degree expansion of $u$ before all virtual nodes hosted by $u$ are members of a proper clean unmatched cluster is $\mathcal{O}(\log^2 N)$.
\label{dlemma:degree_reset}
\end{lemma}
\begin{proof}
By Lemma \ref{dlemma:single_node_increase}, the largest number of nodes any node $v$ will give to node $u$ in a single round is $2 \cdot \log N$.  Furthermore, in order for $u$ to receive $2 \cdot \log N$ nodes from a neighbor $v$, $u$ must host a virtual node whose merge partner is equal to the cluster of the virtual node of $v$.  If $u$ detects virtual nodes without matching tree identifiers, it executes a reset.  If $u$ detects nodes from the same cluster but different levels being connected to a neighbor attempting to merge, $u$ executes a reset.  Therefore, after the first round, at most $2 \cdot \log N$ nodes can be added to $u$'s neighborhood in a single round.  Since, by Lemma \ref{dlemma:reset_to_clean}, all nodes hosted by $u$ are members of a proper clean unmatched cluster in $\mathcal{O}(\log N)$ rounds, the degree expansion of $u$ before all nodes hosted by $u$ are members of the same proper clean unmatched cluster is $\mathcal{O}(\log^2 N)$.

In the initial configuration, it is possible for all neighbors of $u$ to give $u$ $2 \cdot \log N$ neighbors.  In this case, the degree expansion is limited to $\mathcal{O}(\log N)$, since $u$ will reset immediately after receiving these neighbors.
\end{proof}

\begin{lemma}
Let $u$ be a real node such that all virtual nodes hosted by $u$ are members of a proper clean unmatched cluster in configuration $G_i$.  Let $u$'s degree in $G_i$ be $\Delta_u$.  Node $u$'s degree will be at most $\Delta_u + 2 \cdot \log N \cdot (\log N + 1) + 2 \cdot \log N \cdot T(\mathit{lead}))$ until the algorithm terminates, where $T(\mathit{lead})$ is the number of times the virtual nodes hosted by $u$ participate in the leader selection procedure of Algorithm \ref{algo:lead}.
\label{dlemma:degree_clean}
\end{lemma}
\begin{proof}
We consider the three degree-increasing actions that a proper clean unmatched cluster $T$ from configuration $G_i$ may execute.  Consider first the follower selection procedure from Algorithm \ref{dalgo:follow}.  The degree can increase only from node $b$ adding the neighbor $\mathit{leader}$ from cluster $T'$ to the neighborhood of $\mathit{parent}_b$.  Furthermore, this degree increase of one is temporary -- a node deletes an edge to $\mathit{leader}$ after forwarding it, and once the root has the edge, it eventually becomes part of a merge, and either the root of $T$ or the root of $T'$ is deleted.

Next, consider the selection procedure for leaders in Algorithm \ref{algo:lead}.  During the $\mathit{PFC}(\mathit{ConnectFollowers},\bot)$ wave, a virtual node $b$ may receive at most a single neighbor from each child, and after an additional round will retain at most 1 of these edges.  Since a real node $u$ can host at most $2 \cdot \log N$ nodes with children from another host, the degree increase each time a node $u$ participates in the selection procedure for leaders is at most $2 \cdot \log N$.

Finally, consider the degree increase from the merge algorithm.  Assume $b$ and $b'$ are nodes in level $i$ in $T$ and $T'$ (respectively), and suppose $b$ and $b'$ resolve.  Without loss of generality, assume $b'$ is the node which will be deleted (that is, the host of $b'$ is losing part of its range).  The degree of $b$ can only increase by 2 (the children of $b'$).  The degree of $\mathit{host}_b$ may increase if $\mathit{host}_{b'}$ copies some of its virtual nodes to $\mathit{host}_b$.  By Corollary \ref{dcorollary:cluster_degree}, any node $v$ can have at most $2 \cdot \log N$ real neighbors inside cluster $T'$.  Furthermore, there can be exactly one node $b' \in T'$ that updates its successor to $\mathit{host}_b$ at level $i$.  As there are $\log N + 1$ levels, the maximum degree increase during a merge is $2 \cdot \log N \cdot (\log N + 1)$.  Notice that, unlike the degree increase from the selection procedure for leaders, the degree increase from merges is not additive -- the largest node $u$'s degree can be as the result of intra-cluster edges is $2 \cdot \log N$, regardless of how many merges $u$ participates in.  Therefore, while a node's degree may temporarily grow during a merge to $2 \cdot \log N \cdot (\log N + 1)$, after the merge is complete, $u$'s intra-cluster degree is at most $2 \cdot \log N$.
\end{proof}

\end{document}